\documentclass[aip,jmp,amsmath,amssymb,reprint,onecolumn]{revtex4-1}

\usepackage{graphicx}
\usepackage{dcolumn}
\usepackage{bm}

\usepackage[utf8]{inputenc}
\usepackage[T1]{fontenc}

\usepackage{amsthm}
\usepackage{braket}
\usepackage{dsfont}
\usepackage{mathtools}
\usepackage[caption=false]{subfig}

\newcommand{\complex}{\mathbb{C}} 

\newcommand{\Hil}{\mathcal{H}} 
\newcommand{\HilK}{\mathcal{K}} 

\newcommand{\Lin}{\mathcal{L}}
\newcommand{\Lnr}[1]{\Lin(#1)} 
\newcommand{\Bnd}[1]{\mathcal{B} (#1)} 
\newcommand{\Tr}[1]{\mathcal{T} (#1)} 
\newcommand{\St}[1]{\mathcal{S} (#1)} 

\newcommand{\id}{\mathds{1}}
\newcommand{\nul}{0} 

\newcommand{\tr}[1]{\operatorname{tr}\left[#1\right]} 

\newcommand{\obA}{\mathsf{A}}
\newcommand{\obB}{\mathsf{B}}
\newcommand{\obC}{\mathsf{C}}

\newcommand{\Fix}[1]{\operatorname{Fix}(#1)} 
\newcommand{\diag}{\operatorname{diag}}

\theoremstyle{definition}
\newtheorem{thm}{Theorem}

\newtheorem{defn}{Definition}
\newtheorem{cor}{Corollary}
\newtheorem{example}{Example}

\begin{document}
\title{Relation between the state distinction power and disturbance in quantum measurements}
\author{Ikko Hamamura}
 \email{hamamura@nucleng.kyoto-u.ac.jp}

\author{Takayuki Miyadera}
 \email{miyadera@nucleng.kyoto-u.ac.jp}

\affiliation{Department of Nuclear Engineering, Kyoto University, 6158540 Kyoto, Japan}

\date{\today}

\begin{abstract}
The measurement of an informative observable strongly disturbs a quantum state.
We examine the so-called information-disturbance relation by introducing order relations based on the state distinction power of an observable and a variety of non-disturbed observables
with respect to a channel,
and obtain qualitative and quantitative representations.
\end{abstract}

\maketitle

\section{Introduction}\label{sec:introduction}
All the operation pairs are not simultaneously realizable
in quantum theory.
This concept of incompatibility\cite{Heinosaari2016}
is ubiquitous.
For example, the position and momentum of a single particle
cannot be
simultaneously measured;
therefore, these observables are incompatible.
While this example refers to the incompatibility between two observables,
the notion of incompatibility can be extended to general operations.
Each class of operation in the quantum theory is characterized by its outcome space.
An observable has a classical state space as its outcome space.
An operation whose outcome is a
quantum state is called a (quantum) channel.
A pair of operations (possibly belonging to
different classes) is called incompatible
if
there is no device that contains these operations as its parts.

In this study, we discuss the incompatibility between observables and channels.
This type of incompatibility has been extensively studied,
and various conditions have been obtained for ensuring
compatibility between
an observable and a channel \cite{Fuchs1996,Ozawa2003,Ozawa2004,Kretschmann2008,Buscemi2013,Busch2013}.
The compatibility conditions
exhibit an information-disturbance relation;
``strong'' disturbance is inevitable while measuring ``informative'' observables.
By introducing quantitative measures for the informative character of
an observable and the disturbing character of a channel,
the inequalities between an observable and a channel can be
identified;
these inequalities place
concrete limitations on compatible operations.
Recently an author proposed a qualitative (structural) representation
of the information-disturbance relation~\cite{Heinosaari2013}.
An observable $\obA$ (a channel $\Lambda_1$) is defined to be more informative (resp. less disturbing)
than $\obB$ (resp. $\Lambda_2$) if and only if $\obB$ (resp. $\Lambda_2$) can be obtained by $\obA$ (resp. $\Lambda_1$) followed by post-processing,
This definition considers the observable space (channel space)
to be a preordered space that is not linearly ordered.
Further, a relation in terms of the order structure was obtained.
Once a quantitative measure on each space respecting the order is introduced, its corresponding quantitative representation of
the information-disturbance relation is obtained.
Thus, this approach can reveal
the structure of the information-disturbance relations and
justify certain quantitative measures.

While the argument above was entirely based on
the post-processing induced order,
the order is not a unique one that has operational meaning.
The objective of this study
is to examine another qualitative representation of
the information-disturbance relation.
We employ an order in the observable space determined using
 the state distinction power\cite{Busch1989,Heinonen2005}.
Two order relations in the channel space are introduced by focusing on the invariant observables,
which can be naturally interpreted
in terms of the quantum non-demolition measurement \cite{Braginsky1980,Grangier1998,Lupascu2007}.
We demonstrate that these order relations are related to each other by qualitative
and quantitative information-disturbance relations.

The rest of the paper is organized as follows.
In section~\ref{sec:preliminaries} the fundamental
concepts used throughout the paper are introduced.
In section~\ref{sec:preorder},
an order relation in the observable space and two order relations in the channel space are introduced.
In section~\ref{sec:infodist} the relations among the orders
are discussed,
and the qualitative and quantitative relations are obtained.
In addition,
a sequential measurement is discussed
by comparing it with a result within the post-processing order.

\section{Preliminaries}\label{sec:preliminaries}
Let $\Hil$ be a separable (possibly infinite-dimensional)
Hilbert space. Throughout the paper, we consider a quantum system
described by $\Hil$.
We denote by $\Bnd{\Hil}$ the set of bounded operators on $\Hil$.
A normal state is represented by a positive operator with unit trace (i.e., a density operator).
We denote the set of density operators as $\St{\Hil}$.
$\St{\Hil}$
spans the set of trace class operators $\Tr{\Hil}$.
On $\Tr{\Hil}$ one can define a norm $\Vert T\Vert_1 \coloneqq \tr{|T|}$,
which makes this space a Banach space.
Its dual is $\Bnd{\Hil}$ which is also
a Banach space with respect to
$\Vert A\Vert \coloneqq \sup_{T \in \Tr{\Hil}} \frac{|\tr{TA}|}{\Vert T\Vert_1}
=\sup_{|\psi\rangle \neq 0} \frac{\Vert A|\psi\rangle\Vert}{\Vert |\psi\rangle\Vert}$.
In addition to the norm topology induced by
$\Vert \cdot \Vert_1$, one can introduce a weak topology
by the dual structure; $T_n$ converges $T$ weakly
iff $|\tr{T_n A} -\tr{T A}|\to 0$ for all $A\in \Bnd{\Hil}$.
A weak$ ^*$ topology on $\Bnd{\Hil}$ is introduced similarly;
$A_n \to A$ in weak$ ^*$ topology iff $|\tr{TA_n} - \tr{T A} | \to 0$
for all $T \in \Tr{\Hil}$.
In this paper, we treat observables with at most countable outcome sets.
An observable whose outcome set is $\Omega$ is represented by a family of positive operators
$\obA=\set{\obA(x)}_{x\in\Omega}$ satisfying $\sum_{x\in\Omega} \obA(x) =\id$
which is called a positive-operator-valued measure (POVM),
where the summation over infinite $\Omega$ is defined with respect to
the weak$ ^*$ topology.
Hereafter we denote by $\Omega_{\obA}$ the outcome set of an observable $\obA$.
As this set is at most countable, it may be safely identified with a subset of
$\mathcal{N}$.
Suppose that we prepare a state described by a density operator
$\rho$ and measure an observable $\obA$.
Then we obtain each outcome $x \in \Omega_{\obA}$
with probability $\tr{\rho \obA(x)}$.
Thus each observable $\obA$ defines an affine map from the set of
states $\St{\Hil}$ to the set of probability distributions on its outcome set $\Set{\set{p(x)}_{x\in\Omega_{\obA}}| p(x)\geq 0, \sum_x p(x) =1}$.
While an observable is sufficient to describe the
statistical aspects of the classical outcomes of a measurement process, it does not specify any information on the
dynamical change of states.
A state change in general
is described by a map whose output space is a set of quantum states.
We call a map from $\St{\Hil}$ to $\St{\HilK}$ (where
$\mathcal{K}$
is an output Hilbert space)
a channel if it is affine and its natural extension to $\St{\Hil \otimes \complex^{N}}$ has its codomain in $\St{\HilK\otimes\complex^N}$ for each $N\in\mathbb{N}$.
This map, from $\St{\Hil}$ to $\St{\HilK}$,
can be linearly extended to
a trace-preserving completely positive map from $\Tr{\Hil}$ to $\Tr{\HilK}$.
We denote a dual map of $\Lambda$ by $\Lambda^*\colon
\Bnd{\HilK} \to \Bnd{\Hil}$ which is defined by
$\tr{\Lambda(T) B} = \tr{T \Lambda^*(B)}$ for all $T\in \Tr{\Hil}$.
This map $\Lambda^*$ is demonstrated to be weak$ ^*$ continuous.
In this paper, we examine channels whose output system coincides with the input system, i.e., it is assumed that $\HilK = \Hil$.
\par
A complete description of the measurement process is
given by a set of completely positive maps $\set{\mathcal{I}_x }_{x\in \Omega}$
called an instrument,
where $\mathcal{I}_x$ for each $x\in \Omega$ is a
completely positive map from $\Tr{\Hil}$ to itself
such that  $\sum_{x\in \Omega} \mathcal{I}_x$ is a channel,
where the summation is defined with respect to the weak topology.
An instrument $\set{\mathcal{I}_x}_{x\in \Omega}$
gives a POVM $\set{\mathcal{I}_x^*(\id)}_{x\in \Omega}$, where
$\mathcal{I}_x^*$ is a dual map of $\mathcal{I}_x$.
The instrument gives an unconditional state change
by a channel $\sum_x \mathcal{I}_x$.
An observable $\obA=\set{\obA(x)}_{x\in \Omega}$ and a channel $\Lambda$
is called {\em compatible} if and only if
there exists an instrument $\mathcal{I}= \set{\mathcal{I}_x}$
such that $\obA(x) = \mathcal{I}_x^*(\id)$ for each $x\in \Omega$ and
$\sum_{x\in \Omega} \mathcal{I}_x = \Lambda$ hold.
We call a channel $\Lambda$ an $\obA$-channel if and only if
an observable $\obA$ and $\Lambda$ are compatible.

It is well known that each channel can be represented
using an operator sum form (Kraus representation)\cite{Kraus1983}.
There exists, for the given channel $\Lambda$, a family of operators $\set{K_n} \subset\Bnd{\Hil}$ satisfying $\Lambda(\rho) = \sum_{n} K_n \rho K^*_n$,
where each $K_n$ is referred to as
a Kraus operator and satisfies $\sum_n K^*_n K_n = \id$.
While the Kraus representation for a given channel
is not unique, for a channel compatible with an observable $\obA$
there is a convenient Kraus representation
$\Lambda(\rho) = \sum_{x,j}K_{x,j}\rho K_{x,j}^*$
such that $\sum_j K_{x,j}^* K_{x,j} = \obA(x)$ holds for each $x
\in \Omega_{\obA}$.

\section{Preorders of the observables and channels}\label{sec:preorder}

\subsection{Order structure of the
observables induced by the state distinction power}
A pair of states $\rho_1$ and $\rho_2$ is called distinguishable by an observable $\obA$ if and only if the
probability distributions obtained by measuring $\obA$ with respect to
$\rho_1$ and $\rho_2$ are different;
otherwise, the pair is indistinguishable.
Thus $\rho_1$ and $\rho_2$ are distinguishable
by $\obA$ if and only if
there exists an outcome $x \in \Omega$ such that $\tr{\rho_1 \obA(x)}\neq \tr{\rho_2 \obA(x)}$; see Fig~\ref{fig:StateDistinguishability}.
\begin{figure}
\centering
\includegraphics[width=0.7\linewidth]{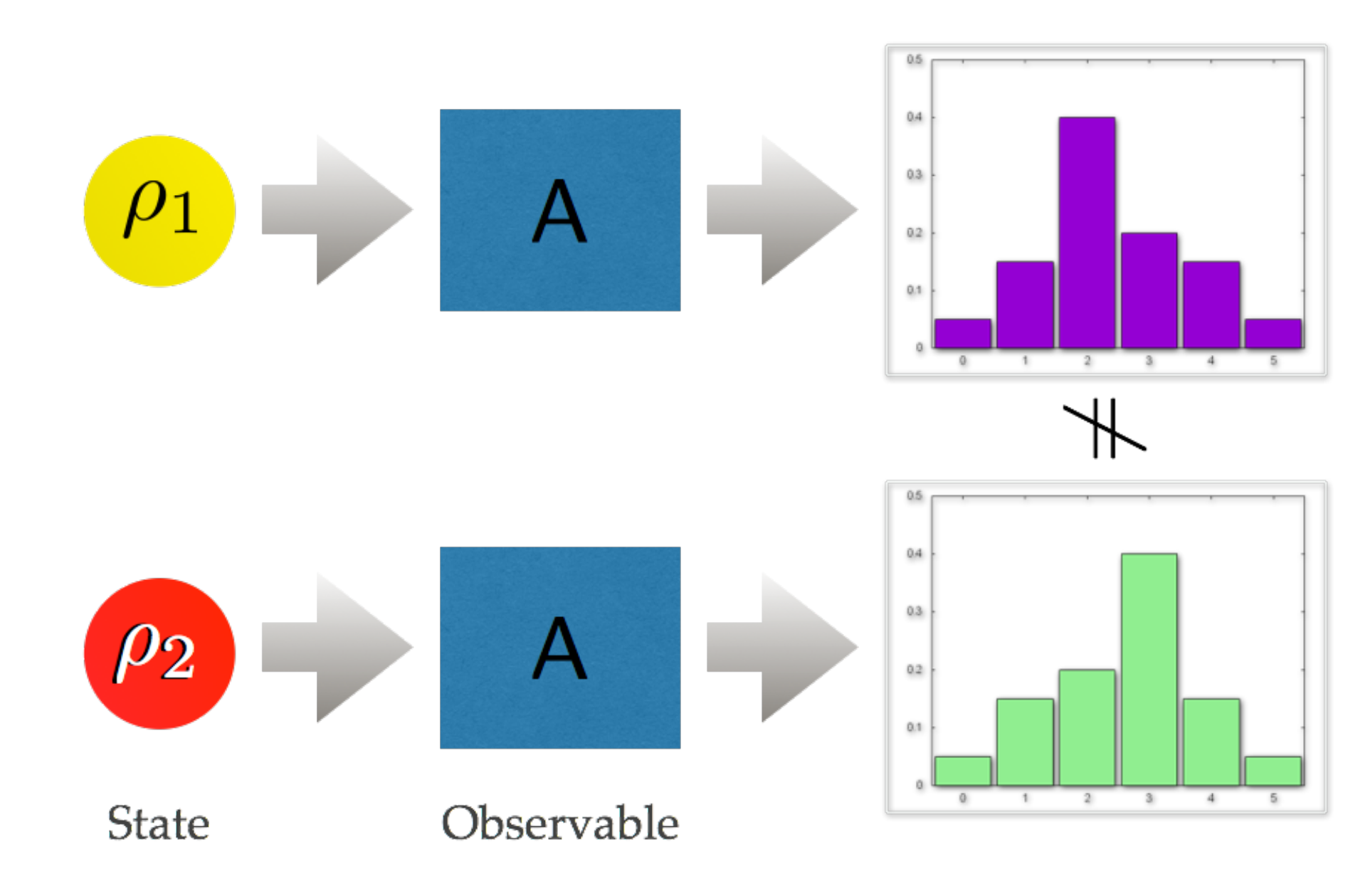}
\caption{
Discriminating a pair of states $\rho_1$ and $\rho_2$ from its probability distribution by measuring an observable $\obA$.}
\label{fig:StateDistinguishability}
\end{figure}
Let us consider a product space $\mathcal{S}(\Hil) \times \mathcal{S}(\Hil)
\subset \Tr{\Hil}\times \Tr{\Hil}$, on which a product topology of weak topology is introduced. Thus, a sequence $(\rho_n, \sigma_n)$ converges to $(\rho, \sigma)$
if and only if for any $(A, B) \in \Bnd{\Hil} \times \Bnd{\Hil}$
both $\lim_n \tr{\rho_n A} = \tr{\rho A}$ and $\lim_n \tr{\sigma_n B}
= \tr{\sigma B}$ hold.
One can introduce a subset $\Sigma_{\obA} \subset \mathcal{S}(\Hil)
\times \mathcal{S}(\Hil)$ for an observable $\obA$ by
\begin{eqnarray*}
\Sigma_{\obA}= \set{(\rho_1, \rho_2)\in \mathcal{S}(\Hil) \times \mathcal{S}(\Hil)|
\rho_1\mbox{ and } \rho_2 \mbox{ are not distinguishable by } \obA}.
\end{eqnarray*}
The following properties are observed to be valid:
\begin{itemize}
\item[(i)] $(\rho_1, \rho_2) \in \Sigma_{\obA}$, then $(\rho_2, \rho_1) \in \Sigma_{\obA}$.
\item[(ii)]$(\rho_1, \rho_2), (\rho_2, \rho_3) \in \Sigma_{\obA}$,
then $(\rho_1, \rho_3)\in \Sigma_{\obA}$.
\item[(iii)]$\Sigma_{\obA}$ is weakly closed.
\end{itemize}
(i) and (ii) are trivial. To show (iii), we investigate $\Sigma_{\obA}^c$.
Suppose $(\rho_1,\rho_2) \in \Sigma_{\obA}^c$.
This implies that there exists $x\in\Omega_{\obA}$ such that
$| \tr{\rho_1 \obA(x)} -\tr{\rho_2 \obA(x)} | >\epsilon >0$.
Then, for all $\rho_1'$ and $\rho_2'$
satisfying $|\tr{\rho_j \obA(x)} -\tr{\rho_j'\obA(x)}| > \epsilon/3$ for $j=1,2$,
  $| \tr{\rho_1' \obA(x)} -\tr{\rho_2' \obA(x)} | >\epsilon/3$.
Thus, there exists an open neighborhood of $(\rho_1, \rho_2)$
contained in $\Sigma_{\obA}^c$.

Based on this distinguishability, a relation between observables called state-distinction power, is introduced\cite{Busch1989,Heinonen2005}.

\begin{defn}
An observable $\obA$ has larger state distinction power than an observable $\obB$ if and only if any pair of states distinguishable by $\obB$ is also distinguishable by $\obA$.
That is, $\Sigma_{\obA} \subseteq \Sigma_{\obB}$ holds.
We denote this relation by $\obA \succsim_i \obB$.
\end{defn}
It is easy to see that the state distinction power is a preorder.
That is, the relation is reflexive ($\obA \precsim_i \obA$ for any $\obA$) and transitive ($\obA \precsim_i  \obB$ and $\obB \precsim_i \obC$ implies $\obA \precsim_i \obC$).
However,
this relation does not satisfy antisymmetric property and is not a partial order (see the argument presented below).

We emphasize that this relation is not total.
That is, there exists a pair of observables $\obA$ and $\obB$ that
 satisfies neither $\obA \precsim_i \obB$ nor $\obB \precsim_i \obA$.
For instance, in the qubit system ($\Hil=\complex^2$), the
sharp observables determined by the
Pauli matrices $\sigma_x$, $\sigma_y$ and $\sigma_z$ have no relation in terms of the state distinction power, e.g. $\sigma_x\not\precsim\sigma_y$ and $\sigma_x\not\succsim\sigma_y$.
Also,
in the qutrit system ($\Hil = \complex^3$), the sharp observables determined by the spin components
\begin{equation}
  S_x = \frac{1}{\sqrt{2}}
\begin{pmatrix}
0 & 1 & 0 \\
1 & 0 & 1 \\
0 & 1 & 0
\end{pmatrix},\quad
S_y = \frac{1}{\sqrt{2}i}
\begin{pmatrix}
  0 & 1  & 0 \\
 -1 & 0  & 1 \\
 0  & -1 & 0
\end{pmatrix},\quad
S_z =
\begin{pmatrix}
  1 & 0 & 0 \\
  0 & 0 & 0 \\
  0 & 0 & -1
\end{pmatrix}.
\end{equation}
are incomparable with each other.

To study the comparability, we introduce the
linear span of a POVM $\obA$ by
\begin{equation}
\mathcal{L}_0(\obA)
= \Set{X \in \Bnd{\Hil} | X = \sum_{x\in \Omega_0} c_x \obA(x), c_x \in \complex, \Omega_0 \subset \Omega \mbox{ with }
|\Omega_0| <\infty
}
\end{equation}
and its weak$ ^*$ closure, $\Lnr{\obA}$.
These sets are vector spaces and self-adjoint; they are operator systems.
\\
For a subset $M \subseteq \Bnd{\Hil}$,
a subset ${}^a M \coloneqq \set{T\in \Tr{\Hil}| \tr{T X} =0, {}^{\forall} X \in M}$
of $\Tr{\Hil}$
is called a preannihilator of $M$. On the other hand,
for a subset $F \subseteq \Tr{\Hil}$,
$F^a \coloneqq \set{X \in \Bnd{\Hil}| \Tr{TX}=0, {}^{\forall} T\in F}$
is called an annihilator of $F$.
For $M \subseteq \Bnd{\Hil}$, as ${}^a M$ is written as
${}^a M = \cap_{X \in M}\set{T \in \Tr{\Hil}|  \tr{T X}=0}$,
${}^a M$ is a closed subspace with respect to both
the weak and norm topology.  Similarly, for $F \subseteq \Tr{\Hil}$,
$F^a$ is a closed subspace with respect to weak$ ^*$ and norm topology.
One can show that for $M \subseteq \Bnd{\Hil}$,
$({}^a M)^a$ is the smallest weak$^ *$ closed subspace containing
$M$. Thus we find $({}^a\mathcal{L}_0(\obA))^a = \Lnr{\obA}$.

States $\rho_1$ and $\rho_2$ are not distinguishable if and only if
$\rho_1- \rho_2$ is in $ ^a \mathcal{L}_0(\obA)$.

The following theorem provides a simple criterion for a pair of observables to be comparable.
\begin{thm}\label{lem:statedistinction}
Let $\obA$ and $\obB$ be observables.
The following conditions are equivalent.
\begin{itemize}
\item[(i)] $\obA \precsim_i \obB$.
\item[(ii)] $\Lnr{\obA} \subseteq \Lnr{\obB}$.
\end{itemize}
\end{thm}
\begin{proof}
(ii) $\Rightarrow$ (i).
Suppose that states $\rho_1$ and $\rho_2$ are distinguishable by
$\obA$; there exists $x$ such that
$\tr{\rho_1 \obA(x)} - \tr{\rho_2 \obA(x)} =\epsilon >0$.
We consider a neighborhood of $\obA(x) \in \Bnd{\Hil}$,
$U_{\epsilon/3} \coloneqq \set{B \in \Bnd{\Hil} |
| \tr{ \rho_j \obA(x)} - \tr{\rho_j B}| <\epsilon/3, j=1,2}$.
Then there exists $B \in \mathcal{L}_0(\obB) \cap U_{\epsilon/3}$,
which satisfies
\begin{eqnarray*}
\epsilon = |\tr{\rho_1 \obA(x)} - \tr{\rho_2 \obA(x)}|
&\leq& |\tr{\rho_1 \obA(x)} - \tr{\rho_1 B}|
+ \tr{\rho_2 \obA(x)} - \tr{\rho_2 B}| +
|\tr{\rho_1 B} - \tr{\rho_2 B}| \\
&<& 2\epsilon /3 + |\tr{\rho_1 B} - \tr{\rho_2 B}|.
\end{eqnarray*}
Thus there exists $B = \sum_{y}^{finite} c_y \obB(y)$ $(c_y \in \mathbb{C})$
such that
\begin{eqnarray*}
|\tr{\rho_1 B} - \tr{\rho_2 B}| > \epsilon/3,
\end{eqnarray*}
which implies that there exists $y$ satisfying
$\tr{\rho_1 \obB(y)} \neq \tr{\rho_2 \obB(y)}$.

(i) $\Rightarrow$ (ii).
Assume that (ii) is not true.
There exists $x\in\Omega_\obA$ such that
$\obA(x) \notin \Lnr{\obB}$.
According to Proposition 14.13 in
a textbook~\cite{BrownPearcy}, there exists
$T\in \Tr{\Hil}$ such that $\tr{T B} =0$ for all $B \in \Lnr{\obB}$ and
$\tr{T\obA(x)} \neq 0$. Because
$\id \in \Lnr{\obB}$, $T\in \Tr{\Hil}$ satisfies
$\tr{T}=0$. By decomposing $T$ into self-adjoint and skew self-adjoint parts,
one can assume that $T$ is self-adjoint. Thus $T$ is written as
$T=T_+ -T_-$, where $T_{\pm}$ are positive operators satisfying
$\tr{T_+}=\tr{T_-} <\infty$. For sufficiently small $\epsilon>0$,
we can find a state $\rho_0$ such that $\rho_1 \coloneqq
\rho_0 + \epsilon T$ is also a state.
These states are not distinguishable by $\obB$ but are distinguishable by $\obA$.
Thus (i) does not hold.
\end{proof}

An observable $\obA$ is called informationally complete if $\tr{\rho_1 \obA(x)} = \tr{\rho_2 \obA(x)}$ for all $x$ implies $\rho_1 = \rho_2$. Thus by definition it is obvious that each informationally complete observable is a maximal element with respect to the preorder of the state distinction power.
In fact, the linear span of an informationally complete observable coincides with $\Bnd{\Hil}$~\cite{Heinosaari2011mathematical}.
\begin{example}
On the qubit system ($\Hil= \mathbb{C}^2$), a symmetric informationally complete (SIC) POVM\cite{Renes2004} $\obA_{\mathrm{SIC}}$ is given by
\begin{alignat*}{2}
\mathsf{A}_{\mathrm{SIC}}(0)&=\frac{1}{4}
\begin{pmatrix}
1+\frac{1}{\sqrt{3}} & \frac{1}{\sqrt{3}}-i\frac{1}{\sqrt{3}} \\
\frac{1}{\sqrt{3}}+i\frac{1}{\sqrt{3}} & 1-\frac{1}{\sqrt{3}}
\end{pmatrix},&
\mathsf{A}_{\mathrm{SIC}}(1)&=\frac{1}{4}
\begin{pmatrix}
1-\frac{1}{\sqrt{3}} & \frac{1}{\sqrt{3}}+i\frac{1}{\sqrt{3}} \\
\frac{1}{\sqrt{3}}-i\frac{1}{\sqrt{3}} & 1+\frac{1}{\sqrt{3}}
\end{pmatrix},\\
\mathsf{A}_{\mathrm{SIC}}(2)&=\frac{1}{4}
\begin{pmatrix}
1-\frac{1}{\sqrt{3}} & -\frac{1}{\sqrt{3}}-i\frac{1}{\sqrt{3}} \\
-\frac{1}{\sqrt{3}}+i\frac{1}{\sqrt{3}} & 1+\frac{1}{\sqrt{3}}
\end{pmatrix},&
\mathsf{A}_{\mathrm{SIC}}(3)&=\frac{1}{4}
\begin{pmatrix}
1+\frac{1}{\sqrt{3}} & -\frac{1}{\sqrt{3}}+i\frac{1}{\sqrt{3}} \\
-\frac{1}{\sqrt{3}}-i\frac{1}{\sqrt{3}} & 1-\frac{1}{\sqrt{3}}
\end{pmatrix}.
\end{alignat*}
It is easy to verify that their linear span coincides with $\Bnd{\Hil}$.
\end{example}
\begin{example}
A SIC-POVM on the qutrit system ($\Hil = \mathbb{C}^3$)
is given by
\begin{alignat*}{3}
\obA^{\mathrm{qutrit}}_{\mathrm{SIC}}(0)&=\frac{1}{6}
\begin{pmatrix}
1 & 1 & 0 \\
1 & 1 & 0 \\
0 & 0 & 0
\end{pmatrix},&\quad
\obA^{\mathrm{qutrit}}_{\mathrm{SIC}}(1)&=\frac{1}{6}
\begin{pmatrix}
1 & \omega^* & 0 \\
\omega & 1 & 0 \\
0 & 0 & 0
\end{pmatrix},&\quad
\obA^{\mathrm{qutrit}}_{\mathrm{SIC}}(2)&=\frac{1}{6}
\begin{pmatrix}
1 & \omega & 0 \\
\omega^* & 1 & 0 \\
0 & 0 & 0
\end{pmatrix}, \\
\obA^{\mathrm{qutrit}}_{\mathrm{SIC}}(3)&=\frac{1}{6}
\begin{pmatrix}
1 & 0 & 1 \\
0 & 0 & 0 \\
1 & 0 & 1
\end{pmatrix},&\quad
\obA^{\mathrm{qutrit}}_{\mathrm{SIC}}(4)&=\frac{1}{6}
\begin{pmatrix}
1 & 0 & \omega^* \\
0 & 0 & 0 \\
\omega & 0 & 1
\end{pmatrix},&\quad
\obA^{\mathrm{qutrit}}_{\mathrm{SIC}}(5)&=\frac{1}{6}
\begin{pmatrix}
1 & 0 & \omega \\
0 & 0 & 0 \\
\omega^* & 0 & 1
\end{pmatrix}, \\
\obA^{\mathrm{qutrit}}_{\mathrm{SIC}}(6)&=\frac{1}{6}
\begin{pmatrix}
0 & 0 & 0 \\
0 & 1 & 1 \\
0 & 1 & 1
\end{pmatrix},&\quad
\obA^{\mathrm{qutrit}}_{\mathrm{SIC}}(7)&=\frac{1}{6}
\begin{pmatrix}
0 & 0 & 0 \\
0 & 1 & \omega^* \\
0 & \omega & 1
\end{pmatrix},&\quad
\obA^{\mathrm{qutrit}}_{\mathrm{SIC}}(8)&=\frac{1}{6}
\begin{pmatrix}
0 & 0 & 0 \\
0 & 1 & \omega \\
0 & \omega^* & 1
\end{pmatrix}
\end{alignat*}
where $\omega=\exp\left(\frac{2\pi}{3}i\right)$.
It can be easily observed
that their linear span coincides with $\Bnd{\complex^3}$.
\end{example}

On the other hand, the minimum element of this preorder is a trivial observable which is written as $\obA(x)=p_x\id$ with $p_x \geq 0$ for each $x \in \Omega_x$.
Any pair of states $\rho_1$ and $\rho_2$ is not distinguishable by the
 trivial observable because $\tr{\rho_1 \obA(x)}= \tr{\rho_2 \obA(x)} = p_x$ holds.
It is easy to see that every minimum element is a trivial observable.


\subsection{Order structure of channels induced by the nondisturbing measurement}

Let us introduce the order structure of channel space in the light of the disturbance property.
A channel describes a state change from $\rho$ to $\Lambda(\rho)$.
Suppose that we examine the change by measuring an observable $\obB=\set{\obB(y)}_{y\in\Omega}$.
We can confirm that $\rho$ and $\Lambda(\rho)$ differ if the probability distributions $\set{\tr{\rho \obB(y)}}$ and $\set{\tr{\Lambda(\rho)\obB(y)}}$ are distinct.
On the other hand, if we select $\obB$ so that $\tr{\rho \obB(y)} = \tr{\Lambda(\rho)\obB(y)}$ holds for all $y \in \Omega$ and for any input state $\rho$,
this $\obB$ is useless to study the state change.
Motivated by this observation, we call
$\Lambda$ nondisturbing\cite{Heinosaari2010} for an observable $\obB$
(or $\obB$ is not disturbed by $\Lambda$)
if $\tr{\rho\obB(y)} = \tr{\Lambda(\rho)\obB(y)}$ holds for any input state $\rho$ and $y\in\Omega$
(see Fig~\ref{fig:NonDisturbingness}).

\begin{figure}
\centering
\includegraphics[width=0.7\linewidth]{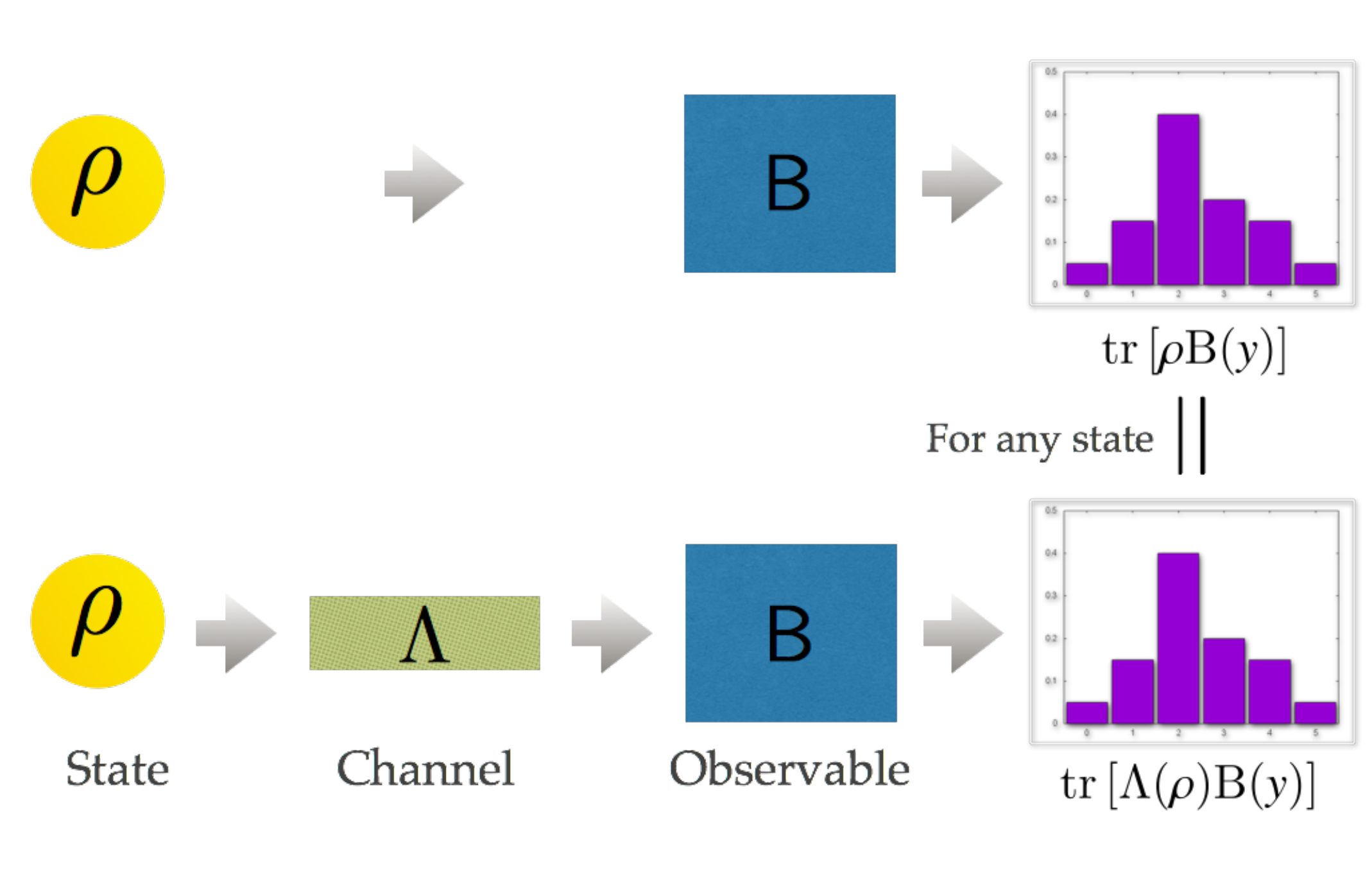}
\caption{Channel $\Lambda$ does not disturb the output from measuring an observable $\obB$}
\label{fig:NonDisturbingness}
\end{figure}

We introduce an order structure on the set of channels.

\begin{defn}
Let $\Lambda_1$ and $\Lambda_2 $ be channels.
If any observable not disturbed by $\Lambda_1$ is
also not disturbed by $\Lambda_2$,
we call $\Lambda_2$ less disturbing than $\Lambda_1$ and denote by $\Lambda_1 \precsim_f \Lambda_2$.
\end{defn}

This binary relation is a preorder.
The Heisenberg picture is useful to characterize this preorder.
Channel $\Lambda$ has a dual description (i.e., Heisenberg picture) $\Lambda^*: \Bnd{\Hil} \to \Bnd{\Hil}$ defined by $\tr{\Lambda(\rho) X} = \tr{\rho \Lambda^*(X)}$ for any $\rho$.
The set of fixed points of the map $\Lambda^*$ is defined by
\begin{equation}
\Fix{\Lambda^*} = \Set{X\in \Bnd{\Hil}| \Lambda^*(X)=X},
\end{equation}
which forms a weakly closed operator system in $\Bnd{\Hil}$.
In fact, for any $X \in \Fix{\Lambda^*}$, as $\Lambda^*(X^*)
= \Lambda^*(X)^* = X^*$ holds, $\Fix{\Lambda^*}$
is self-adjoint. In addition, as $\Fix{\Lambda}=
\ker(\mathbf{id} - \Lambda^*)$
holds, it is a weakly closed set.

\begin{thm}
Let $\Lambda_1$ and $\Lambda_2$ be the channels.
The following conditions are equivalent.
\begin{itemize}
\item[(i)] $\Lambda_1 \precsim_f \Lambda_2$.
\item[(ii)] $\Fix{\Lambda_1^*} \subset \Fix{\Lambda_2^*}$.
\end{itemize}
\end{thm}
\begin{proof}
(i) $\Rightarrow$ (ii).
Assume that there exists $X \in \Fix{\Lambda_1^*}
\cap \Fix{\Lambda_2^*}^c$. One can show that there exists
$Y \in \Fix{\Lambda_1^*} \cap \Fix{\Lambda_2^*}$
such that $\nul \leq Y \leq \id$.
In fact, any $X$ can be decomposed into
$X= \mathrm{Re} X + i\mathrm{Im} X$, where $\mathrm{Re} X \coloneqq \frac{X+X^*}{2}$
and $\mathrm{Im} X \coloneqq \frac{X-X^*}{2i}$. As $\Fix{\Lambda_1^*}$
is self-adjoint,
$\mathrm{Re} X, \mathrm{Im} X \in \Fix{\Lambda_1^*}$ holds.
In addition, $X\notin \Fix{\Lambda_2^*}$ implies that
at least one of them is not in $\Fix{\Lambda_2^*}$.
Let us assume that $\mathrm{Re} X \notin \Fix{\Lambda_2^*}$.
One can choose $a,b \in \mathbb{R}$ so that
$\nul \leq Y \coloneqq a \mathrm{Re} X + b \id \leq \id$.
It is easy to see $Y \in \Fix{\Lambda_1^*}
\cap \Fix{\Lambda_2^*}^c$.
We introduce a POVM $\set{Y, \id - Y}$.
This POVM is disturbed by $\Lambda_2$, but not
by $\Lambda_1$. It contradicts.
\\
(ii) $\Rightarrow$ (i): Trivial.
\end{proof}

Thus there exists the greatest element in the channel space.
The greatest element is an identity channel $\mathbf{id}$.
One can easily verify that $\Fix{\mathbf{id}^*}=\Bnd{\Hil}$ holds.
On the other hand, the least element does not exist.
One of the minimal elements is a channel whose output state does not depend on the input state.
That is, $\Lambda$ has a fixed $\rho_0$ such that $\Lambda(\rho) = \rho_0$ holds for any $\rho$.
Other examples of minimal elements are presented later.

In the aforementioned setting,
we can employ an arbitrary POVM
to confirm if an input state
$\rho$ and an output state $\Lambda(\rho)$ differ.
On the other hand, we can restrict the observables to be measured after channel $\Lambda$ to be sharp ones.
Therefore we introduce the following notion.
\begin{defn}
Let $\Lambda_1$ and $\Lambda_2$ be the channels.
If any PVM that is not disturbed by $\Lambda_1$
is also not disturbed by $\Lambda_2$, we denote
the relation by $\Lambda_1 \precsim_S \Lambda_2$.
\end{defn}
It is easy to see that
$\Lambda_1 \precsim_f \Lambda_2$ implies
$\Lambda_1 \precsim_S \Lambda_2$.
In addition, we obtain the following theorem.
\begin{thm}
For channels $\Lambda_1$ and $\Lambda_2$ the following
conditions are equivalent:
\begin{itemize}
\item[(i)] $\Lambda_1 \precsim_S \Lambda_2$
\item[(ii)] $\Fix{\Lambda_1^*} \cap P(\mathcal{H})
  \subseteq \Fix{\Lambda_2^*} \cap P(\mathcal{H})$,
where $P(\mathcal{H})$ is the set of all projection operators.
\end{itemize}
\end{thm}

\section{Information-disturbance relation}
\label{sec:infodist}
\subsection{Qualitative formulation}

In this section, we study the qualitative information-disturbance
relation based on the preorder structures.
Relationship between orders $\precsim_i$ and
$\precsim_S$ is investigated.
\begin{thm}
Let $\obA$ be an observable. There exists a channel $\Lambda^{\obA}$
satisfying that for any $\obA$-channel $\Lambda$ the following relation holds:
\begin{eqnarray*}
\Lambda \precsim_S \Lambda^{\obA}.
\end{eqnarray*}
That is, there exists a maximum channel with respect to the order
$\precsim_S$ in the set of $\obA$-channels.
\end{thm}
\begin{proof}
  Suppose $P \in \Fix{\Lambda^*} \cap P(\mathcal{H})$.
Let us introduce a Kraus representation of
$\Lambda$ by $\Lambda(T) =\sum_{x,j} K_{x,j} T K_{x,j}^*$
satisfying $\sum_j K_{x,j}^* K_{x,j}= \obA(x)$ for each $x$.
In accordance with the argument in Lindblad~\cite{Lindblad1999},
we consider
\begin{eqnarray*}
\Lambda^*(P) - \Lambda^*(P) P - P \Lambda^*(P) +P
= \sum_{j,x} [K_{j,x}, P]^* [K_{j,x}, P]=0.
\end{eqnarray*}
Thus we find $[K_{j,x}, P]=[K_{j,x}^*, P]= 0$ for all $j$.
Furthemore, we obtain
\begin{eqnarray*}
[\obA(x), P] = \sum_{j} [ K_{j,x}^* K_{j,x}, P]
= \sum_j [K_{j,x}^*, P] K_{j,x} + K_{j,x}^* [K_{j,x}, P] =0.
\end{eqnarray*}
Thus we can conclude that
\begin{eqnarray*}
\Fix{\Lambda^*} \cap P(\mathcal{H}) \subseteq \mathcal{L}(\obA)'
\cap P(\mathcal{H}).
\end{eqnarray*}
Let us introduce $\Lambda^{\obA}$ by
$\Lambda^{\obA}(T) = \sum_x \sqrt{\obA(x)} T \sqrt{\obA(x)}$,
a L\"{u}ders channel.
It is obvious that each $P \in \mathcal{L}(\obA)'
\cap P(\mathcal{H})$ satisfies $P \in \Fix{\Lambda^*} \cap P(\mathcal{H})$.
Thus the equality $\Fix{\Lambda^{\obA *}}\cap P(\mathcal{H})
= \mathcal{L}(\obA)' \cap P(\mathcal{H})$ holds.
\end{proof}
In contrast to the post-processing order version,
the set of $\obA$-channels is not an ideal with respect to
$\precsim_S$ as demonstrated by the following example.
\begin{example}
Let us consider a qubit system and a channel $\Lambda(T)
= \sigma_z T \sigma_z$. $\Fix{\Lambda}$ is spanned by $\id$
and $\sigma_z$, and is thus nontrivial. However, this channel is
not compatible with any nontrivial observable including $\sigma_z$.
\end{example}
The following theorem which can be referred to as
a qualitative information-disturbance relation, is easy to obtain.
\begin{thm}
Let $\obA$ and $\obB$ be observables satisfying $\obA \precsim_i \obB$.
Suppose that $\Lambda$ is a channel compatible with $\obB$.
Then there exists a channel $\Gamma$ compatible with $\obA$
such that
$\Lambda \precsim_S \Gamma$ holds.
\end{thm}
\begin{proof}
Consider $\Gamma = \Lambda^{\obA}$.
\end{proof}
Suppose that the $\obA$-channel $\Lambda$ has a
nontrivial projection $P \in \Fix{\Lambda^*} \cap P(\mathcal{H})$.
The following argument yields an explicit construction of a pair of indistinguishable states of $\obA$.

Now there exist unit vectors $|\psi_1\rangle$ and $|\psi_2\rangle$
satisfying $P| \psi_1\rangle =|\psi_1\rangle$ and $P|\psi_2 \rangle =0$.
A pair of states $\Lambda( |\psi_1\rangle \langle \psi_1|) $ and
$\Lambda( |\psi_2 \rangle \langle \psi_2|)$ is perfectly distinguishable by
a PVM $\set{P, \id -P}$. In fact,
$\tr{\Lambda(|\psi_1\rangle \langle \psi_1|)P}=1$
and $\tr{\Lambda(|\psi_2 \rangle \langle \psi_2)P}=0$ hold.
According to the paper\cite{Heinosaari2013},
any $\obA$-channel is written in the following form:
\begin{eqnarray*}
\Lambda(T) = \mathcal{E} \left(\sum_x \hat{\obA}(x)V  T V^*\hat{\obA}(x)\right),
\end{eqnarray*}
where $(\set{\hat{\obA}(x)}_x, V)$ is a Naimark extension of
the POVM $\obA$ and $\mathcal{E}$ is a channel.
That is, $V\colon \Hil \to \mathcal{K}$ is
an isometry to a larger Hilbert space $\mathcal{K}$ and
$\set{\hat{\obA}(x)}_x$ is a PVM on $\mathcal{K}$ satisfying
$V^* \hat{\obA}(x) V= \obA(x)$ for each $x\in \Omega_{\obA}$,
and $\mathcal{E} \colon\mathcal{T}(\mathcal{K}) \to
\mathcal{T}(\mathcal{H})$ is a channel.
Thus one can conclude that
a pair of subnormalized states
$\mathcal{E}(\hat{\obA}(x)V |\psi_1\rangle \langle \psi_1 |
V^* \hat{\obA}(x) )$ and
$\mathcal{E} (\hat{\obA}(x)V|\psi_2 \rangle \langle \psi_2|V^*
\hat{\obA}(x))$ is perfectly distinguishable.
Because any channel does not increase the distinguishability,
we observe that a pair of subnormalized states
$\hat{\obA}(x)V |\psi_1\rangle \langle \psi_1| V^*\hat{\obA}(x)$
and $\hat{\obA}(x)V |\psi_2 \rangle \langle \psi_2|V^* \hat{\obA}(x)$
is perfectly distinguishable. This demonstrates that
\begin{eqnarray*}
\langle \psi_1| \obA(x) |\psi_2\rangle
=0
\end{eqnarray*}
 holds for each $x$.
Therefore we can conclude that for
$|\psi\rangle = c_1 |\psi_1 \rangle + c_2 |\psi_2\rangle$
with $|c_1|^2 + |c_2|^2=1$, a pair of states
$|\psi\rangle \langle \psi|$ and $|c_1|^2 |\psi_1\rangle
\langle \psi_1| + |c_2|^2 |\psi_2\rangle \langle \psi_2 |$
is not dinstinguishable by $\obA$.
In particular, the following theorem holds.
\begin{thm}
Let $\obA$ be an informationally complete observable
and $\Lambda$ be an $\obA$-channel.
Then it holds that for any channel $\Gamma$,
\begin{eqnarray*}
\Lambda \precsim_S \Gamma.
\end{eqnarray*}
That is, $\Lambda$ is a minimum channel.
\end{thm}

Let us remark that the converse is not true.
The following example shows a minimal channel that is not compatible with any informationally complete observable.

\begin{example}
\label{Ex3}
Let us suppose that $\dim \Hil =N\geq 3$, which has a standard basis $\set{\ket{n} }_{n=0}^{N-1}$.
We define unitary matrices $U= \sum_{n} e^{i \frac{2\pi n}{N}}|n \rangle \langle n|$ and $V= \sum_n |n+1\rangle \langle n|$, where $|N\rangle $ is identified with $|0\rangle$. It is easy to see that $\set{U, V}' = \complex \id$ holds.
We consider a channel $\Lambda(\rho) = \frac{1}{2}U\rho U^* + \frac{1}{2} V\rho V^*$.
This has a faithful fixed point $\frac{\id}{N}$.
Based on the previous discussion, we conclude that this channel is a minimal element in the channel space.
On the other hand, every POVM element of an observable compatible with this channel is spanned by $\set{\id, U^*V, V^*U}$.
Thus for any $\obA$ compatible with $\Lambda$, $\dim \mathcal{L}(\obA) \leq 3$ holds.
We can conclude that no informationally complete observable is compatible.
\end{example}

It is unclear whether there exists any simple relation between
the orders $\precsim_i$ and $\precsim_S$.
In the following section, we investigate their relation using a
 quantitative formulation.

\subsection{Quantitative inequality}\label{sec:quantitative}

Hereafter, we assume $\Hil$ to be finite dimensional.
To systematically investigate the relation, we introduce a subclass of
channels.
A channel is said to satisfy {\em property F}
 if and only if it has a faithful invariant state.
That is, $\Lambda$ with property F has a state $\rho_0$ whose kernel is $\set{0}$ such that $\Lambda (\rho_0)= \rho_0$.
Fixed point $\Fix{\Lambda^*}$ of this channel forms a subalgebra of $\Bnd{\Hil}$ \cite{Lindblad1999}.
In fact a Kraus decomposition $\Lambda^*(\cdot)=\sum_{j,x} K_{j,x}^* \cdot K_{j,x}$ allows its fixed point $\Fix{\Lambda^*}$ to be
\begin{equation}
\Fix{\Lambda^*}=\Set{K_{j,x},K_{j,x}^*}'_{j,x}=\Set{X\in\Bnd{\Hil}|[X,K_{j,x}]=[X,K_{j,x}^*]=0 \mbox{ for all }j,x}
\end{equation}
It is used to obtain
\begin{eqnarray}
\Fix{\Lambda^*} \subseteq \mathcal{L}(\obA)'.
\label{lemma3}
\end{eqnarray}

It should be noted that $\Lnr{\obA}$ and $\Fix{\Lambda^*}$ are vector spaces.
A natural quantity used to measure the size of a complex vector space is its dimension.
We obtain the following theorem.
\begin{thm}
\label{thm:quantitative}
For an observable $\obA$ and a channel $\Lambda^*$ that are compatible,
it holds that
\begin{equation}
\label{eq:inequality}
\dim \Lnr{\obA} + \dim \Fix{\Lambda^*} \leq (\dim\Hil)^2 + 1.
\end{equation}
Furthermore, equality is attained if and only if an observable $\obA$ is informationally complete or a channel $\Lambda^*$ is the identity channel.
\end{thm}

\begin{proof}
We first assume that channel $\Lambda$ satisfies
property F.
We prove $\Delta\coloneqq(\mathrm{RHS})-(\mathrm{LHS})\geq 0$.
\begin{align*}
\Delta
&= (\dim \Hil)^2 + 1 - \dim \Lnr{\obA} - \dim \Fix{\Lambda^*} \\
&\geq (\dim \Hil)^2 + 1 - \dim \Lnr{\obA} - \dim \Lnr{\obA}' \quad\mbox{((\ref{lemma3}) is used)}\\
&\geq (\dim \Hil)^2 + 1 - \dim \Lnr{\obA}'' - \dim \Lnr{\obA}'.
\end{align*}

As $\set{\obA(x)}'$ is a finite von Neumann algebra,
it can be represented as
\begin{equation*}
\set{\obA(x)}' = \bigoplus_{n=1,2,\ldots,N} \Bnd{\Hil_n}\otimes\id_{\HilK_n},
\end{equation*}
where $\Hil_n\otimes\HilK_n$ are orthogonal subspaces satisfying $\bigoplus\Hil_n\otimes\HilK_n=\Hil$\cite{Takesaki2002}.
Therefore $\set{\obA(x)}'' = \bigoplus_{n=1,2,\ldots,N} \id_{\Hil_n}\otimes\Bnd{\HilK_n}$ and $\set{\obA(x)}''\cap\set{\obA(x)}'=\bigoplus_{n}\complex P_n$ hold where $P_n$ is a projection on $\Hil_n\otimes\HilK_n$.
We obtain
\begin{align*}
\Delta
  &\geq  (\dim \Hil)^2 + 1 - \dim \bigoplus_{n=1}^N\id_{\Hil_n}\otimes\Bnd{\HilK_n} - \dim \bigoplus_{n=1}^N\Bnd{\Hil_n}\otimes\id_{\HilK_n}\\
&=(\dim\Hil)^2+1-\sum_{n=1}^N(\dim\HilK_n)^2-\sum_{n=1}^N(\dim\Hil_n)^2\\
  &=\left(\sum_{n=1}^N(\dim\Hil_n)\times(\dim\HilK_n)\right)^2+1-\sum_{n=1}^N((\dim\Hil_n)^2+(\dim\HilK_n)^2)\quad\mbox{($\because\Hil=\oplus_{n=1}^N\Hil_n\otimes\HilK_n$)}\\
&\geq \sum_{n=1}^N(\dim\Hil_n)^2(\dim\HilK_n)^2+N(N-1)+1-\sum_{n=1}^N((\dim\Hil_n)^2+(\dim\HilK_n)^2).
\end{align*}
Since an equality $N(N-1)+1=N+(N-1)^2$ holds,
\begin{align}
\Delta
&\geq \sum_{n=1}^N(\dim\Hil_n)^2(\dim\HilK_n)^2+N+(N-1)^2-\sum_{n=1}^N((\dim\Hil_n)^2+(\dim\HilK_n)^2)\notag\\
&= (N-1)^2 + \sum_{n=1}^N ((\dim\Hil_n)^2-1)((\dim\HilK_n)^2-1)\notag\\
&\geq (N-1)^2\label{eq:inequalitydim}\\
&\geq 0.\label{eq:inequalityN-1}
\end{align}

Further, we discuss a case in which
$\Lambda$ does not satisfy property F.
We use a support projection operator $P_\Lambda$ introduced by Lindblad\cite{Lindblad1999}.
This projection operator $P_\Lambda$ is defined by the smallest projection satisfying $\tr{\rho P_\Lambda}=1$ for all $\rho\in\St{\Hil}\cap\Fix{\Lambda}$.
We assume $\dim P_\Lambda \Hil = n < d =\dim \Hil$.
It can be observed that map $\Lambda (P_{\Lambda} \cdot P_{\Lambda})$ is well-defined as a channel on $\St{P_{\Lambda}\Hil}$.
As demonstrated by Lindblad, there exists an invariant state $\rho_0$ for $\Lambda$ whose support projection coincides with $P_{\Lambda}\Hil$.
By using this state $\rho_0$, we observe that $\tr{(\id - P_{\Lambda}) \Lambda(P_{\Lambda} \rho_0 P_{\Lambda})} = \tr{(\id - P_{\Lambda}) \Lambda(\rho_0)}=0$.
As any state $\sigma$ on $P_{\Lambda} \Hil$ is dominated by $\rho_0$, we conclude that
$\Lambda(P_{\Lambda} \sigma P_{\Lambda}) = P_{\Lambda} \Lambda(P_{\Lambda} \sigma P_{\Lambda})P_{\Lambda}$ holds.

Let $\obA$ be an observable compatible with $\Lambda$.
Then an observable $P_{\Lambda} \obA P_{\Lambda}$ on a subspace $P_{\Lambda}\Hil$ is compatible with a channel $\Lambda(P_{\Lambda} \cdot P_{\Lambda})$.
We apply the first part of the present proof to $\Lambda(P_{\Lambda}\cdot
P_{\Lambda})$ to obtain
\begin{align}
\dim P_{\Lambda} \Lnr{\obA}P_{\Lambda}
+ \dim \Fix{\Lambda(P_{\Lambda} \cdot P_{\Lambda})}
&=
\dim P_{\Lambda} \Lnr{\obA}P_{\Lambda}
+ \dim \Fix{\Lambda}\notag\\
&\leq n^2 +1.
\label{siki1}
\end{align}
As $\Lnr{\obA}$ is decomposed as $\Lnr{\obA}= P_{\Lambda} \Lnr{\obA} P_{\Lambda} + P_{\Lambda} \Lnr{\obA} (\id - P_{\Lambda}) +(\id - P_{\Lambda} )\Lnr{\obA} P_{\Lambda} +(\id - P_{\Lambda}) \Lnr{\obA} (\id - P_{\Lambda})$,
the dimension of $\Lnr{\obA}$ is bounded by $\dim P_{\Lambda} \Lnr{\obA} P_{\Lambda} + (d^2 -n^2)$.
Combining this bound with (\ref{siki1}) ends the proof of the first half
of the claim.

It is easy to show that equality is attained if $\obA$ is an informationally complete or $\Lambda$ is the identity channel.
We show the opposite implication.
In the support projection $P_\Lambda$, as equality of inequality (\ref{eq:inequalityN-1}) holds, $N=1$ is required.
To obtain equality (\ref{eq:inequalitydim}),
$\dim\Hil_1=1$ or $\dim\HilK_1=1$ must hold.
This implies that $\dim\HilK_1=\dim\Hil$ or $\dim\Hil_1=\dim\Hil$ holds; thus the observable $\obA$ is an informationally complete or
the channel $\Lambda$ is the identity channel.
\end{proof}

The following is obtained immediately.
\begin{cor}
Let $\obA$ be an informationally compete observable.
Then any channel $\Lambda$ compatible with $\obA$
is minimal with respect to the order $\precsim_i$;
$\Fix{\Lambda^*} = \mathbb{C}\id$ holds.
\end{cor}

\begin{figure}
  \centering
  \subfloat[Qubit system]{\label{fig:quantitative2}
  \includegraphics[width=.48\linewidth]{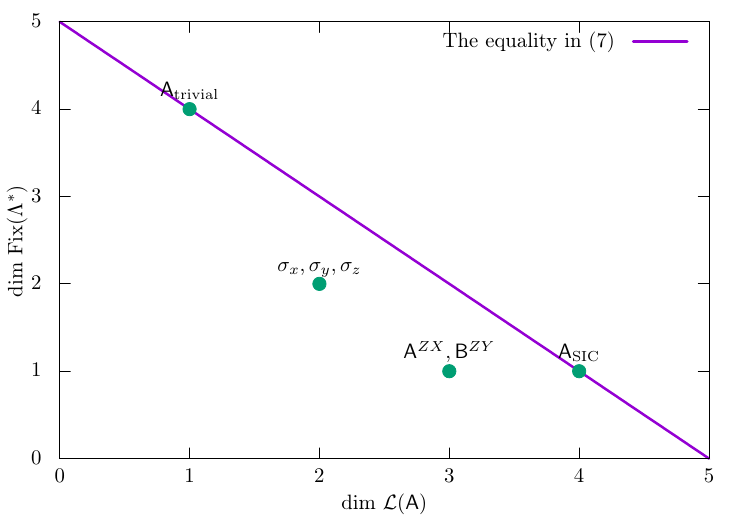}
  }
  \subfloat[Qutrit system]{\label{fig:quantitative3}
 \includegraphics[width=.48\linewidth]{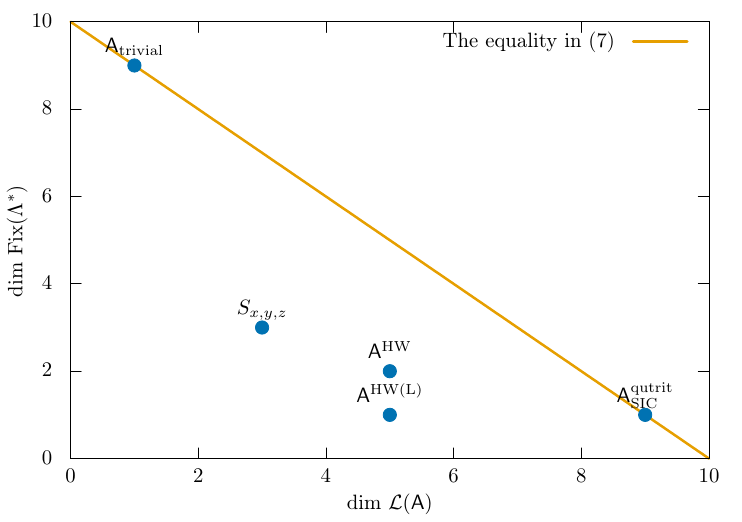}
  }
  \caption{Relation between the quantitative state distinction power
    $\dim\Lnr{\obA}$ and the non-disturbingness $\dim\Fix{\Lambda^*}$.}
\label{fig:quantitative}
\end{figure}

We plot $\dim \Lnr{\obA}$ and $\dim \Fix{\Lambda}$ for $\obA^\mathrm{HW}$ defined by the following Kraus
operators and for the L\"{u}ders channel in Figure~\ref{fig:quantitative}.
$\obA^\mathrm{HW}$ represents a channel defined by Kraus operators
\begin{alignat*}{3}
  K_1 &= \frac{1}{2}
  \begin{pmatrix}
    \sqrt{2} & 0 & -1 \\
    0        & 0 & 0 \\
    0       & 0 & 0
  \end{pmatrix},&\qquad
  K_2 &= \frac{1}{10}
  \begin{pmatrix}
    0 & 0 & 0 \\
    0 & -\sqrt{10} & 2\sqrt{10} \\
    0 & 0 & 0
  \end{pmatrix},&\qquad
  K_3 &= \frac{1}{2}
  \begin{pmatrix}
    0 & 0 & 0 \\
    0 & \sqrt{2} & 0 \\
    0 & 0 & 0
  \end{pmatrix},\\
  K_4 &= \frac{1}{10}
  \begin{pmatrix}
    0 & 0 & 0 \\
    0 & 2\sqrt{10} & \sqrt{10} \\
    0 & 0 & 0
  \end{pmatrix},&\qquad
  K_5 &= \frac{1}{2}
  \begin{pmatrix}
    \sqrt{2} & 0 & 1 \\
    0 & 0 & 0 \\
    0 & 0 & 0
  \end{pmatrix}
\end{alignat*}
and an observable defined by  $\obA^\mathrm{HW}(i) = K_i^*K_i$($i=1,\cdots,5$).
This channel does not satisfy property F\cite{Heinosaari2010}.
$\obA^{\mathrm{HW(L)}}$ is L\"{u}ders channel of $\obA^{\mathrm{HW}}$.
The fixed point set of $\obA^{\mathrm{HW(L)}}$ is spanned by
$\diag(1,1,1)$, whose dimension is $1$.

\subsection{Sequential measurements}

In a previous paper\cite{Heinosaari2015}, the problem of sequential measurement was discussed concerning the post-processing order structure.
It was shown that for each observable $\obA$ there exists
a measurement process which does not spoil the joint measurability.
In particular, after such a measurement process whose channel is
called a universal channel, one can extract information
on an arbitrary maximum observable which is more informative than
$\obA$ by measuring an observable.

Inspired by this observation, we consider the following problem.
Suppose that we first measure an observable $\obA$ and then measure an observable $\obB$ which was not disturbed by the first measurement.
This sequential measurement provides a joint measurement of $\obA$ and $\obB$.
We examine whether it is
possible to measure $\obA$ so that this joint measurement provides  maximal state distinction power.
In other words,
we study if the joint observable $\obA$ and $\obB$ can be informationally complete.
A trivial $\obA$ provides a trivial example.
As the measurement of $\obA$ does not cause any disturbance,
one can subsequently measure any informationally complete observable $\obB$.
A less trivial example is as follows.
We assume that $\Hil$ has a tensor product structure as $\Hil=\Hil_1\otimes\Hil_2$.
Suppose that $\obA$ is an informationally complete observable on $\Hil_1$, and $\obB$ is another informationally complete observable on $\Hil_2$.
Then $\obA\otimes\obB$ becomes an informationally complete observable on $\Hil$.
We can construct a channel compatible with $\obA$ acting only on $\Hil_1$,
which naturally does not disturb $\obB$.
This preservation of state distinction power does not hold in general.

\begin{thm}\label{thm:sequential-not-ic}
Let $\obA$ be an observable such that $\set{\obA(x)}''\cap\set{\obA(x)}'$ is not trivial.
That is, this algebra is strictly larger than $\complex\id$.
We measure $\obA$ by a compatible channel $\Lambda$.
Let $\obB$ be an observable satisfying $\obB(y)\in\Fix{\Lambda^*}$ for all $y\in\Omega_\obB$.
The subsequent measurement of $\obB$ provides a joint measurement of $\obA$ and $\obB$.
This joint measurement cannot be informationally complete.
\end{thm}
\begin{proof}
Let $\Lambda$ be a channel compatible with $\obA$.
We denote its Kraus representation $\Lambda(\cdot)=\sum_x\sum_i K_{x,i}^* \cdot K_{x,i}$.
$\Fix{\Lambda^*}$ is contained in $\set{\obA(x)}'$.

As $\set{\obA(x)}'$ is a finite von Neumann algebra,
it can be represented as
\begin{equation*}
\set{\obA(x)}' = \bigoplus_{n=1,2,\ldots,N} \Bnd{\Hil_n}\otimes\id_{\HilK_n},
\end{equation*}
where $\Hil_n\otimes\HilK_n$ are orthogonal subspaces satisfying $\bigoplus\Hil_n\otimes\HilK_n=\Hil$\cite{Takesaki2002}.
Therefore $\set{\obA(x)}'' = \bigoplus_{n=1,2,\ldots,N} \id_{\Hil_n}\otimes\Bnd{\HilK_n}$ and $\set{\obA(x)}''\cap\set{\obA(x)}'=\bigoplus_{n}\complex P_n$ hold where $P_n$ is a projection on $\Hil_n\otimes\HilK_n$.
The assumption of its nontriviality is equivalent to the condition $N\geq 2$.

Let us construct states that cannot be distinguished by the joint observable $\obA$ and $\obB$.
We define a unitary operator $U$ by $U=\bigoplus_n e^{i\theta_n}P_n$ satisfying $\theta_n\neq\theta_m \pmod{2\pi}$ for all $n\neq m$.
The states we are now considering are a some arbitrary state $\ket{\phi_0}$ and $\ket{\phi_1}=U\ket{\phi_0}$.
These states cannot be distinguished by the joint measurement.
In fact,
\begin{align*}
\Pr(\obA=x,\obB=y|\phi_0)
  &=\Braket{\phi_0|\sum_i K_{x,i}^* \obB(y)K_{x,i}\phi_0} \\
  &=\Braket{\phi_0|\obA(x)\obB(y) \phi_0} \\
\Pr(\obA=x,\obB=y|\phi_1)
  &=\Braket{\phi_0|U^*\sum_i K_{x,i}^* \obB(y)K_{x,i}U \phi_0} \\
  &=\Braket{\phi_0|U^*\obA(x) \obB(y)U\phi_0} \\
&=\Braket{\phi_0|\obA(x)\obB(y)\phi_0}.
\end{align*}
Thus they correspond to each other.
However, $\ket{\phi_0}$ and $\ket{\phi_1}$ are not identical in general.
\end{proof}

\section{discussion}
\label{sec:discussion}
In this paper, we studied the compatibility between observables and
channels.
We investigated the state distinction power which characterizes the quality of the observables. This notion introduces a preorder structure completely characterized by the operator system $\mathcal{L}(\obA)$.
In the channel space, based upon a discussion on the nondisturbing channel, we introduced two preorders $\precsim_f$ and $\precsim_S$.
The order structures are completely characterized by
the fixed point set $\Fix{\Lambda^*}$ of each channel $\Lambda$.
We further investigated the relations between the
orders $\Lambda_i$ and
$\Lambda_S$. It was shown that
the set of $\obA$-channels contains a maximum
channel with respect to $\Lambda_S$.
In contrast to the result for the post-processing order structure,
the set of $\obA$-channels is not an ideal in general.
In addition, a quantitative relation between $\Lambda_i$ and
$\Lambda_S$ was obtained.

 It is interesting to compare our results with those obtained for the post-processing order structure.
 While both the results exhibit
 their own information-disturbance relations,
 they are different from each other in several aspects.
 For instance, while the latter offers a complete characterization of
 the $\obA$-channels as a principal ideal, the former does not.
 In addition, their maximum channels in $\obA$-channels
 do not coincide in general.

 Extending the present analysis to channel-channel compatibility
 is a future problem, on which we hope to report elsewhere.

\section*{Acknowledgement}
IH acknowledges support by Grant-in-Aid for JSPS Research Fellow (JP18J10310).
TM acknowledges financial support from JSPS
(KAKENHI Grant Number 15K04998).

\bibliography{state_distinction-disturbance}

\end{document}